\documentclass[letterpaper, 10 pt, conference,preprint]{ieeeconf} 

\IEEEoverridecommandlockouts

\usepackage[utf8]{inputenc}
\usepackage[T1]{fontenc}
\usepackage{amsfonts}
\usepackage{amssymb}  
\usepackage{amsmath} 

\usepackage{mathtools}
\usepackage{amscd}
\usepackage{color}
\usepackage[english]{babel}

\usepackage{transparent}
\usepackage{tikz}

\usepackage{amsthm}
\usepackage{mdframed}

\usepackage[noadjust]{cite} 

\usepackage{pgfplots}
\pgfplotsset{compat=newest}
\newlength\fwidth

\newcommand\copyrighttext{%
	\footnotesize \copyright 2019 IEEE. Personal use of this material is permitted.  Permission from IEEE must be obtained for all other uses, in any current or future media, including reprinting/republishing this material for advertising or promotional purposes, creating new collective works, for resale or redistribution to servers or lists, or reuse of any copyrighted component of this work in other works.}
\newcommand\copyrightnotice{%
	\begin{tikzpicture}[remember picture,overlay]
	\node[anchor=south,yshift=10pt] at (current page.south) {\fbox{\parbox{\dimexpr\textwidth-\fboxsep-\fboxrule\relax}{\copyrighttext}}};
	\end{tikzpicture}%
}

\graphicspath{{figures/}{Figures/}{Abbildungen/}}

\newcommand{\barx}[1]{\mkern 1.5mu\overline{\mkern-1.5mu#1\mkern-1.5mu}\mkern 1.5mu}
\newcommand{\bs}[1]{\boldsymbol{#1}}
\newcommand{\sgn}{\text{sgn}}

\newtheoremstyle{test}
	{}
	{}
	{}
	{}
	{\itshape}
	{}
	{.5em}
	{\thmname{#1} \thmnumber{#2}:}
\theoremstyle{test}
\newtheorem{lemma}{Lemma}
\newtheorem{remark}[lemma]{Remark}
\newtheorem{assumption}[lemma]{Assumption}
\newtheorem{problem}{Problem}

\title{\LARGE \bf
	Economic model predictive control for snake robot locomotion	
}

\author{Marko Nonhoff, Philipp N. Köhler, Anna M. Kohl, Kristin Y. Pettersen, and Frank Allg\"ower
	\thanks{Marko Nonhoff is with the Institute of Automatic Control, Leibniz University Hannover, 30167 Hannover, Germany, Philipp N. Köhler and Frank Allgöwer are with the Institute for Systems Theory and Automatic Control, University of Stuttgart, 70550 Stuttgart, Germany.
		Anna Kohl and Kristin Y. Pettersen are with Centre for Autonomous Marine Operations and Systems (NTNU AMOS), Department of Engineering Cybernetics, Norwegian University of Science and Technology, 7491 Trondheim, Norway.}
}

\begin{document}
	
	\maketitle
	\thispagestyle{empty}
	\pagestyle{empty}
	\copyrightnotice
	
	\begin{abstract}
		
	In this work, the control of snake robot loco-motion via economic model predictive control (MPC) is studied. Only very few examples of applications of MPC to snake robots exist and rigorous proofs for recursive feasibility and convergence are missing. We propose an economic MPC algorithm that maximizes the robot's forward velocity and integrates the choice of the gait pattern into the closed loop. We show recursive feasibility of the MPC optimization problem, where some of the developed techniques are also applicable for the analysis of a more general class of system. Besides, we provide performance results and illustrate the achieved performance by numerical simulations. We thereby show that the economic MPC algorithm outperforms a standard lateral undulation controller and achieves constraint satisfaction. Surprisingly, a gait pattern different to lateral undulation results from the optimization.	
	\end{abstract}

	\section{Introduction}
	
	
	There has been active research on the mechanisms and control of snake robot locomotion over the last decades. As opposed to robots which move in a more traditional way, for example wheeled or legged robots, snake robots carry the potential to not only move in cluttered and irregular environments, but also make use of obstacles to aid in their locomotion. Since the snake robot is a robot manipulator arm that can also locomote, it has a wide range of applications, including firefighting as well as search and rescue tasks. Correspondingly, the scope of research activities ranges from movement in tight spaces on land to subsea operations. 
	An overview over previous literature on modelling, analysis, control and application of snake robots is found in~\cite{Liljeback2012Overview,Liljeback2012, PETTERSEN2017}.
	
	The complex dynamics of snake robot locomotion, possible environment interactions, and the presence of constraints on both the input and the states make model predictive control (MPC) a promising approach to control the motion of snake robots. MPC is a control method that solves a finite-horizon optimal control problem at every sampling time instance and applies the first part of the optimal input, see for example \cite{Rawlings2009}.
	In many applications, setpoint stabilization may not be the primary control objective, but rather optimization of some general performance criterion. For this reason, so called economic MPC was developed, which allows to minimize a general performance criterion \cite{MUller2017_EconomicandDistributed}. Economic MPC has been studied in detail and convergence, stability, and performance results are available for the closed loop (see, e.g., \cite{MUller2017_EconomicandDistributed, Ellis2014, angeli2012,Muller2016}). 
	It is a well-known property of economic MPC that it can lead to periodic behavior which commonly occurs in snake robot locomotion. Nevertheless, only very few results exist for applications of MPC to the control of snake robot locomotion. In \cite{Marafioti2014}, MPC is employed for path following of a snake robot in terms of optimizing the parameters of a predefined gait pattern. However, the gait pattern was chosen offline and no rigorous proofs for recursive feasibility of the MPC optimization problem were given for the proposed controller. 
	
	The work at hand develops a theoretical basis and provides first results on utilizing Model Predictive Control techniques for the control of snake robot locomotion without a predefined gait pattern. Therefore, we consider a simplified model and assume the snake robot moves on a flat and unbounded surface.
	In particular, by directly computing the input signals for each joint of the snake robot, the choice of the gait pattern is integrated into the closed loop, as opposed to existing approaches which, e.g., employ feedback controllers to track a predefined gait pattern \cite{Liljeback2012Overview,Liljeback2012, PETTERSEN2017} or central pattern generators~\cite{IJSPEERT2008}.

	Thereby, \pubidadjcol we enable the snake robot to adapt its motion to, e.g., a changing environment, faults or changing performance criteria given by an altered cost function. For simplicity, we focus on maximizing the forward velocity of the snake robot throughout this work, since this yields a simple value function with an intuitive physical interpretation and is a reasonable objective of locomotion. Nevertheless, the central ideas are applicable to more \mbox{complex} stage costs as well.
	
	This paper is structured as follows: In Section 2, we introduce the mathematical model of the snake robot used in this work and review the lateral undulation gait pattern and a corresponding controller from the literature. \mbox{Section 3} introduces the economic MPC scheme, which is investigated \mbox{analytically} regarding recursive feasibility and performance in Section 4. Subsequently, numerical simulations are shown in Section 5, and concluding remarks are given in Section 6.
	
	\textit{Notation:} Let $\mathbb{I}_{[a,b]}$ denote the set of all integers in the interval $[a,b] \subset \mathbb{R}$, and let $\mathbb{I}_{>a}$ denote the set of all integers larger than $a \in \mathbb{R}$. $\lceil x \rceil$ is the ceiling function, i.e., \mbox{$\lceil x \rceil = \min \{ n \in \mathbb{I} | n \geq x \}$,} and $\sgn (x)$ represents the signum function.

	We consider a discrete-time nonlinear system
	\begin{align} \label{eq:defgensys}
	x(t+1) &= f(x(t),u(t)), \hspace{1cm} x(0) = x_0,
	\end{align} 
	where $f:\mathbb{X} \times \mathbb{U} \rightarrow \mathbb{R}^n$, $x(t) \in \mathbb{X} \subseteq \mathbb{R}^n$ and \mbox{$u(t) \in \mathbb{U} \subseteq \mathbb{R}^m$} are the system dynamics, the state and the control input at time $t \in \mathbb{I}_{\geq 0}$, and $\mathbb{X}$ and $\mathbb{U}$ denote the state and input constraint sets, respectively. 
	The solution of system~\eqref{eq:defgensys} for a control sequence $\bs u = ( u(0), \ldots, u(K-1) ) \in \mathbb{U}^{K}$ starting at the initial value $x_0 \in \mathbb{X}$ is denoted by $\bs x^{\bs u}(t, x_0)$, $t=0,\dots,K$, which is abbreviated as $\bs x^{\bs u}(t)$ whenever $x_0$ is clear from the context.
	
	\section{Problem Setup}
	
	In this section, we present a mathematical model of the snake robot which will be used for controller design, analysis and simulations in the remainder of this work. Modelling a snake robot in full detail yields a dynamical system too complicated for contoller design. For this reason, we will use a simplified model throughout this work, specifically developed for  controller design and analysis. A detailed derivation of this simplified model can be found in \cite{Liljeback2012,Liljeback2013}. The main idea is to describe the robot by a serial connection of translational instead of revolute joints connecting the links of the snake robot, since analysis shows that it is the transversal motion of the links that is significant for forward motion. A schematic representation of the modelling approach is provided in Figure \ref{fig:snakemodel}.
	
	\begin{figure}[t]
		\centering
		\def\svgwidth{.45\textwidth}
		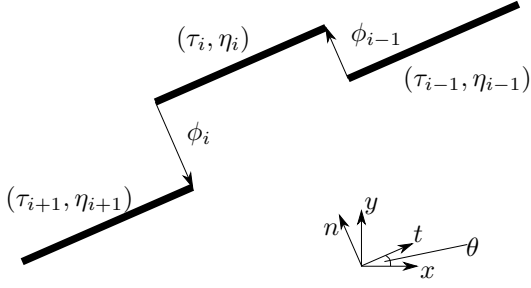
		\caption{Schematic representation of the snake model and the used coordinate systems.}
		\label{fig:snakemodel}
	\end{figure}
	
	We consider a planar snake robot consisting of $N_l$ links, all with the same mass $m$ and length $l$, which are interconnected by $N_l-1$ translational joints. The center of mass of each link is located at its center point. The robot moves on a horizontal and flat surface and is driven by $N_l -1$ actuators, one at each joint. 
	
	First, we define a set of matrices which will be used extensively in this work.
	\begin{flalign*}
	A &= \left[ \begin{smallmatrix} 1 & 1 &   &  \\
	& \ddots & \ddots &    \\
	&   & 1 & 1 \end{smallmatrix} \right]
	\in \mathbb{R}^{(N_l-1) \times N_l}, \\
	D &= \left[ \begin{smallmatrix} 1 & -1 &   &  \\
	& \ddots & \ddots &    \\
	&   & 1 & -1 \end{smallmatrix} \right]
	\in \mathbb{R}^{(N_l-1) \times N_l}, \\
	e &= [1, \ldots, 1]^T \in \mathbb{R}^{N_l}, \qquad \barx{e} = [1, \ldots, 1]^T \in \mathbb{R}^{N_l-1}, \\
	\barx{D} &= D^T \left( D D^T \right) ^{-1} \in \mathbb{R}^{N_l \times (N_l-1)}.
	\end{flalign*}
	
	The matrices $A$ and $D$ represent an addition and subtraction, respectively, of adjacent elements of a vector. We assume that the snake robot is subject to an anisotropic viscous ground friction force. More specifically, the ground friction normal to the link is greater than the ground friction parallel to the link. Both for biological snakes and for snake robots, this is the essential feature enabling them to propel forward. We denote the friction coefficient in normal direction by \mbox{$c_n \in \mathbb{R}_{>0}$} and in tangential direction by \mbox{$c_t \in \mathbb{R}_{>0}$.} Furthermore, we define a propulsion coefficient $c_p$ as
	\begin{align*}
	c_p &= \frac{c_n-c_t}{2l}.
	\end{align*}
	
	We are now ready to state the complete simplified model. As opposed to \cite{Liljeback2012,Liljeback2013}, we give a discretized version, since our MPC scheme will be formulated in discrete-time with sampling time $T_s$.
	\begin{subequations}
		\begin{align}
		\phi(t+1) &= \phi(t) + T_s v_{\phi}(t) \label{eq:SysDynBegin} \\ 
		\theta(t+1) &= \theta(t) + T_s v_{\theta}(t) \\
		p_{x}(t+1) &= p_{x}(t) + T_s \big( v_{t}(t) \cos(\theta(t)) - v_{n}(t) \sin(\theta(t)) \big) \\
		p_{y}(t+1) &= p_{y}(t) + T_s \big( v_{t}(t) \sin(\theta(t)) + v_{n}(t) \cos(\theta(t)) \big) \\
		v_{\phi}(t+1) &= v_{\phi}(t) + T_s u(t) \label{eq:SysDynvphi} \\
		v_{\theta} (t+1) &= v_{\theta}(t) + T_s \left( -\lambda_1 v_{\theta}(t) + \frac{\lambda_2}{N_l-1} v_{t}(t) {\barx{e}}^T \phi(t) \right) \\
		v_{t}(t+1) &= v_{t}(t) + T_s \left(-\frac{ c_t}{m} v_{t}(t) + \frac{2 c_p}{N_lm} v_{n}(t) {\barx{e}}^T \phi(t) \right. \nonumber \\ & \quad \left. -\frac{c_p}{N_lm} \phi^T(t) A {\barx{D}} v_{\phi}(t) \right) \label{eq:SysDynvt} \\ 
		v_{n}(t+1) &= v_{n}(t) + T_s \left( - \frac{c_n}{m} v_{n}(t) + \frac{2 c_p}{N_lm} v_{t}(t) {\barx{e}}^T \phi(t) \right), \label{eq:SysDynEnd}
		\end{align}
		\label{eq:SysDyn}
	\end{subequations}
	where $\phi(t) \in \mathbb{R}^{N_l-1}$ and $v_\phi(t) \in \mathbb{R}^{N_l-1}$ denote the joint distances and velocities, $(p_x(t),p_y(t)) \in \mathbb{R}^2$ represents the position of the snake robot's center of mass in the global frame, $v_t(t) \in \mathbb{R}$ and $v_n(t) \in \mathbb{R}$ are the tangential and normal velocities of the center of mass in the $t-n$ frame, \mbox{$\theta(t) \in \mathbb{R}$} denotes the orientation and $v_\theta(t) \in \mathbb{R}$ the snake robot's rotational velocity. The parameters $\lambda_1 \in \mathbb{R}$ and $\lambda_2 \in \mathbb{R}$ are empirical constants which describe the rotational dynamics. A reasonable choice of parameters for this model is presented in \cite{Liljeback2012}.
	
	\begin{remark} Note that we directly consider the input $u$ obtained through an input transformation as shown in \cite{Liljeback2012,Liljeback2013}, leading to the simple dynamics \eqref{eq:SysDynvphi}. \end{remark}
	
	Due to mechanical restrictions of the snake robot and the model only being valid for joint distances $\phi(t)$ which are sufficiently small \cite{Liljeback2012}, state constraints on the joint distances $\phi(t)$ and the corresponding velocities $v_\phi(t)$ and input constraints need to be respected. These constraints are given as box constraints, hence,
	{
	\setlength{\belowdisplayskip}{.5ex}
	\begin{align}
		\begin{split}\mathbb{X} = \{ &x(t) \in \mathbb{R}^{2N_l + 4} \,\big|\, \phi_i(t) \in [-\phi_{\max}, \phi_{\max}], \\ &v_{\phi,i}(t) \in [-v_{\phi,\max}, v_{\phi,\max}] \hspace{.2cm}\forall i \in \mathbb{I}_{[1,N_l-1]}  \}, \end{split} \label{const:states} \\
		\mathbb{U} = \{ &u(t) \in \mathbb{R}^{N_l-1} | u_i(t) \in [-u_{\max}, u_{\max}] \hspace{.2cm} \forall i \in \mathbb{I}_{[1,N_l-1]} \}, \label{const:input}
	\end{align}
}
	where
	\begin{align} \label{eq:defXsnake}
	x(t) &= [\phi(t), \theta(t), p_x(t), p_y(t), v_\phi(t), v_\theta(t), v_t(t), v_n(t)]^T, 
	\end{align}
	and $\phi_{\max} \in \mathbb{R}_{>0}$, $v_{\phi,\max} \in \mathbb{R}_{>0}$ and $u_{\max} \in \mathbb{R}_{>0}$. We assume that the remaining states are unconstrained.
	
	In \cite{Liljeback2012,Liljeback2013}, a controller was presented which steers the joint distances $\phi(t)$ to a given reference trajectory, defined by the gait pattern lateral undulation (LU). This gait pattern propagates a body wave from head to tail of the snake robot and is defined as
	\begin{align} \label{eq:DefLU}
	\phi_{\mathrm{LU},i}(t) &= \alpha \sin \left( \omega t + (i-1) \delta \right),
	\end{align}
	where $i \in \mathbb{I}_{[1, N_l-1]}$, and $\alpha,\omega,\delta \in \mathbb{R}$ are constant parameters. This gait pattern was studied in detail, e.g., in \cite{Liljeback2012,Saito2002}.
	
	The corresponding lateral undulation controller is given by
	\begin{align}
	\begin{split}
	u(t) &= u_{\mathrm{ref}}(t) + k_d ({v}_{\phi,\mathrm{ref}}(t) - {v}_\phi(t)) \\ &\quad + k_p (\phi_{\mathrm{ref}}(t) - \phi(t)),
	\end{split} \label{eq:defcontLU}
	\end{align}
	with
	\begin{align*}
	\phi_{\mathrm{ref},i}(t) &= \phi_{\mathrm{LU},i}(t), \\
	v_{\phi,\mathrm{ref},i}(t) &= \frac{\mathrm{d}}{\mathrm{d}t} \phi_{\mathrm{ref},i}(t), & u_{\mathrm{ref},i}(t) = \frac{\mathrm{d}^2}{\mathrm{d}t^2} \phi_{\mathrm{ref},i}(t).
	\end{align*}
	This controller was proven to exponentially stabilize the reference gait pattern for the snake robot model \eqref{eq:SysDyn}. Furthermore, in \cite{Liljeback2012,Liljeback2013} it was shown that the average forward velocity converges exponentially fast to a velocity which depends on the parameters describing the gait pattern $\alpha, \omega, \delta$ and $\phi_0$, the friction coefficients $c_n$, $c_t$ and $c_p$, and the number of links $N_l$.
	
	\section{Economic MPC scheme for snake robot locomotion} \label{sec:EMPCpreliminaries}
	
	As mentioned in the introduction, we aim at maximizing the forward velocity of the snake robot as a reasonable objective, and in order to arrive at a simple cost function with an intuitive physical interpretation. We therefore employ $-v_t(t)$ as the cost to be minimized by the MPC optimization problem. However, other choices are possible. For instance, in order to limit the energy consumption of the snake robot, a term $\gamma u^T(k|t) u(k|t)$ with $\gamma \in \mathbb{R}_{>0}$ could be added (cf. Section \ref{sec:sim}). All our results on recursive feasibility provided in the next section are independent of the choice of the specific cost function. However, our results on the performance of the closed loop would need to be adjusted for a modified cost function.
	
	Next, we state our proposed economic MPC algorithm for snake robot locomotion. At each time step $t$, given an initial value $x(t) \in \mathbb{X}$, the following MPC optimization problem is solved.
	{\begin{problem}{(Economic MPC optimization problem)} \label{prob:EMPC}
	\begin{align*} 
	\min_{\bs u(t) \in \mathbb{U}^{N_p}} \hspace{.3cm} &J(x(t), \bs u(t)) = -\sum_{k=0}^{N_p} v_t(k|t) \\
	\text{s.t.} \hspace{.6cm} 
	&x(0|t) = x(t) \\
	&x(k+1|t) = f(x(k|t), u(k|t))\\
	&u(k|t) \in \mathbb{U} \subseteq \mathbb{R}^{N_l-1} \\
	&\hspace{2cm} k = 0, \ldots, N_p-1\\
	&x(k|t) \in \mathbb{X} \subseteq \mathbb{R}^{2N_l+4} \\
	&\hspace{2cm}  k = 0, \ldots, N_p.
	\end{align*}
	\end{problem}}
	In Problem \ref{prob:EMPC}, the snake robot's dynamics $f(x(k|t), u(k|t))$ and states $x(t)$ are given by \eqref{eq:SysDyn} and \eqref{eq:defXsnake}, respectively. We denote the solution to this MPC optimization problem by \mbox{$\bs u^*(t) = ( u^*(0|t), \ldots, u^*(N_p-1|t) ) \in \mathbb{U}^{N_p}$} and the corresponding predicted trajectory by \mbox{$\bs x^*(t) = ( x^*(0|t), \ldots, x^*(N_p|t) ) \in \mathbb{X}^{N_p+1}$}, where the asteriks signify optimality. At each timestep, the first element of the optimal input sequence $\bs u^*(t)$ is applied to the system, hence, the control input is given by $u_{\mathrm{MPC}}(t) = u^*(0|t)$.
	
	\section{Analysis of the closed loop}
	\subsection{Recursive feasibility}
	
	Recursive feasibility is a fundamental property of MPC algorithms, which is required to apply the MPC scheme to a system. Recursive feasibility means that, if a feasible solution to the MPC optimization problem at time $t=t_0$ exists, then there is a feasible solution to the problem for every $t > t_0$. Hence, recursive feasibility establishes that the control law given by the above algorithm is defined at all times. In the following, we present a sufficient condition for recursive feasibility of Problem~\ref{prob:EMPC}, which can be achieved without additional terminal costs or constraints, facilitating implementation and, more importantly, not requiring any a priori knowledge of, e.g., a desirable gait pattern. We note that the results of this section can be applied to a wider range of problems of similar structure, i.e., systems with box constraints on states adhering to double integrator dynamics. This is discussed in more detail below.
	
	The main idea of the following derivations is to provide a candidate solution to the MPC optimization Problem~\ref{prob:EMPC} based on a feasible solution from the previous time step in order to certify that the resulting control input is always defined. 
	
	Given the current state $x(t)$ (including  $v_\phi(t)$), let the input sequence $\bs u^c(t) \in \mathbb{U}^{N_p}$ and the $i$-th entry (referring to the $i$-th joint) of its $k$-th element $\bs u_i^c(k|t)$, $i \in \mathbb{I}_{[1, N_l-1]}$, be defined by the feedback law
	\begin{align} \label{eq:Defuc}
	u_i^c(k|t) &= 
	\begin{cases}
	-\sgn \left(v_{\phi,i} (k|t)\right) u_\mathrm{max} & \text{if } |v_{\phi,i}(k|t)| > T_s u_\mathrm{max} \\
	-\frac{v_{\phi,i}(k|t)}{T_s} & \text{else}
	\end{cases}
	\end{align}
	and let
	\begin{align} \label{eq:defb}
	b = \lceil\frac{v_{\phi,\mathrm{max}}}{T_s u_\mathrm{max}}\rceil.
	\end{align}
	Hence, it is always possible to steer $v_\phi(t)$ to zero in $b$ time steps. Let 
	\begin{align*}
	\bs \phi^*(t) &= ( \phi^*(0|t), \ldots, \phi^*(N_p|t)), \\
	\bs v_\phi^*(t) &= ( v_\phi^*(0|t), \ldots, v_\phi^*(N_p|t))
	\end{align*}
	denote the optimal trajectories of the joint distances and velocities at time $t$, respectively. We define the candidate input sequence $\mathbf{\tilde{u}}(t)$ by shifting the previously optimal solution for the first $N_p-b-1$ time steps and extending it by $\bs u^c(t)$, i.e.,
	\begin{equation} \label{eq:Defudi}
	\begin{split}
	& \mathbf{\tilde{u}}(t+1) = ( u^*(1|t), \ldots, u^*(N_p-b-1|t), \\ &\qquad u^c(N_p-b-1|t+1), \ldots, u^c(N_p-1|t+1)).
	\end{split}
	\end{equation}
	\begin{remark} We use the candidate input $\mathbf{\tilde{u}}(t)$ only for feasibility analysis, and do not intend to actually apply it to the snake robot. When applied to a snake robot, the forward velocity achieved with this input would be undesirable, since it steers the joint velocities $v_\phi(t)$ to zero. Therefore, no acceleration can be achieved by body shape changes and the robot is decelerated due to friction. \end{remark}
	
	The main idea behind defining this input is that it indeed steers the joint velocities to zero, as will be shown by the next lemma. We thereby ensure that the states $v_\phi(t)$ and $\phi(t)$ remain feasible thereafter and only the transient phase of the candidate input sequence $\bs u^c(t)$ and the corresponding state trajectories remain to be analyzed. In order to prove the main result of this section, we first examine the response of the joint velocities $v_\phi^{\bs u^c(t)}(t)$ when actuated by the candidate input sequence $\bs u^c(t)$. Namely, we show that the absolute value of $v_{\phi,i}^{\bs u^c}(t)$ is decreasing and its sign does not change if $|v_{\phi,i}(t)|>T_su_{\max}$. Loosely speaking, this means that the joint velocities are steered towards zero.
	\begin{lemma} \label{lemma:aux}
	Let $\bs u^c(t)$ be defined by (\ref{eq:Defuc}). Then it holds that
	\begin{itemize}
		\item[(i)] $|v_{\phi,i}(k|t)| \geq |v_{\phi,i}^{\bs u_i^c}(k+1|t)|$ for all $i \in \mathbb{I}_{[1,N_l-1]}$ and $k \in \mathbb{I}_{[0,N_p-1]}$.
		\item[(ii)] Moreover, if $|v_{\phi,i}(k|t)| > T_su_{\max}$, then $\sgn(v_{\phi,i}(k|t)) = \sgn\left(v_{\phi,i}^{u_i^c}(k+1|t)\right)$ for all $i \in \mathbb{I}_{[1,N_l-1]}$ and $k \in \mathbb{I}_{[0,N_p-1]}$.
	\end{itemize}
	\end{lemma}

	\begin{proof} (i) Analyzing the dynamics of $v_{\phi,i}(k|t)$ for some \mbox{$i \in \mathbb{I}_{[1,N_p-1]}$} and $k \in \mathbb{I}_{[0,N_p-1]}$ given by \eqref{eq:SysDynvphi} when applying $u_i^c(k|t)$ yields
	\begin{align*}
	|v_{\phi,i}^{u_i^c}(k+1|t)| &= |v_{\phi,i} (k|t) + T_s u^c_i(k|t)|.
	\intertext{First, assume that $|v_{\phi,i}(k|t)| > T_s u_{\max}$, which gives}
	|v_{\phi,i}^{u_i^c}(k+1|t)|| &= |v_{\phi,i}(k|t)-T_s \sgn\left( v_{\phi,i}(k|t)\right) u_{\max}| \\
	&< |v_{\phi,i}(k|t)|.
	\intertext{Second, assume otherwise $|v_{\phi,i}(k|t)| \leq T_s u_{\max}$, which gives}
	|v_{\phi,i}^{u_i^c}(k+1|t)|| &= |v_{\phi,i}(k|t)-T_s \frac{v_{\phi,i}(k|t)}{T_s}| = 0 \\
	&\leq |v_{\phi,i}(k|t)|.
	\end{align*}
	Combining these two results yields the desired inequality $|v_{\phi,i}(k|t)| \geq |v_{\phi,i}^{u_i^c}(k+1|t)|$.
	
	(ii) Due to  $|v_{\phi,i}(k|t)| > T_su_{\max}$ and the definition of the input $u_i^c(k|t)$ in \eqref{eq:Defuc} it holds that \mbox{$u_i^c(k|t) = -\sgn \left(v_{\phi,i}(k|t)\right)u_{\max}$.} Therefore,
	\begin{align*}
	&\sgn \left(v_{\phi,i}^{u_i^c}(k+1|t)\right) = \sgn \left(v_{\phi,i}(k|t) + T_s u_i^c(k|t) \right) \\
	= &\sgn \left( \sgn \left( v_{\phi,i}(k|t) \right) \left( |v_{\phi,i}(k|t)| - T_s u_{\max} \right) \right) \\
	= &\sgn \left(v_{\phi,i}(k|t) \right). \qedhere
	\end{align*}
	\end{proof}
	
	Note that Lemma \ref{lemma:aux}(i) already implies recursive \linebreak feasibility of $v_{\phi,i}(k|t)$, i.e., if $|v_{\phi,i}(k|t)| \leq v_{\phi,\max}$ then \mbox{$|v_{\phi,i}^{u_i^c}(k+1|t)| \leq |v_{\phi,i}(k|t)| \leq v_{\phi,\max}$.}
	
	Next, we state the main result of this section. It provides a sufficient condition for recursive feasibility of Problem \ref{prob:EMPC}. The only required assumption is existence of an initially feasible input at time $t=t_0$ as is standard in MPC feasibility analysis.
	\begin{lemma} \label{lemma:recfeas}
	Suppose that there exists an initially feasible solution $\bs u^*(t)$ to Problem \ref{prob:EMPC} at time $t = t_0$. If $N_p>b$, then there exists a feasible solution to Problem \ref{prob:EMPC} for all $t \in \mathbb{I}_{>t_0}$. \end{lemma} 

	\begin{proof} We assume that Problem \ref{prob:EMPC} was feasible at time $t$. Then, it has to be shown that the input constraint \eqref{const:input} and both state constraints \eqref{const:states} can be satisfied at time $t+1$. As common in MPC, this is done by showing feasibility of a suboptimal candidate trajectory, which in our case results from application of the candidate input sequence \eqref{eq:Defudi}.
	
	First, we examine the input constraint. For the first \mbox{$N_p-b-1$} time steps the candidate input sequence \mbox{$\mathbf{\tilde{u}}(t+1)$} is feasible since it is the same as the shifted optimal input ${\bs u}^*(t)$ which is feasible by assumption.
	
	For the remaining time steps, the candidate input sequence is feasible, since $u^c_i(k|t)$ is either defined by 
	\begin{align*}
	|u_i^c(k|t+1)| &= |-\sgn \left(v_{\phi,i} (k|t+1)\right) u_{\max}| \leq u_{\max}
	\intertext{or by}
	|u_i^c(k|t+1)| &= |-\frac{v_{\phi,i}(k|t+1)}{T_s}| \leq \frac{|T_s u_{\max}|}{T_s} = u_{\max}.
	\end{align*}
	It follows that the candidate input trajectory satisfies the input constraints for all $k \in \mathbb{I}_{[0,N_p-1]}$.
	
	Second, we investigate the response of the joint velocities $\bs v_\phi(k|t+1)$ to the candidate input sequence.
	For the first $N_p-b$ steps it holds that the constraint on $v_\phi^{\mathbf{\tilde{u}}}(k|t+1)$ is satisfied, since the solution of Problem \ref{prob:EMPC} at time $t$ is feasible and, by the construction of the candidate input sequence \eqref{eq:Defudi}, the resulting trajectories are the same.
	For the remaining time steps, constraint satisfaction is provided by Lemma \ref{lemma:aux}(i).
	
	Finally, we examine the response of the joint distances $\phi(k|t+1)$ to the candidate input. For notational convenience, we define $\barx{b} = N_p - b$. Again, by the construction of the candidate input sequence \eqref{eq:Defudi}, the trajectory is the same as the optimal trajectory for the first $N_p-b+1$ time steps.
	\begin{align}
	&\phi^{\mathbf{\tilde{u}}}(k|t+1) = \phi^* (k+1|t), \qquad k \in \mathbb{I}_{[0,N_p-b]}. \nonumber
	\intertext{For the remaining time steps, we obtain $p \in \mathbb{I}_{[1,b]}$ and}
	&\phi^{\mathbf{\tilde{u}}}(\barx{b}+p|t+1) \nonumber \\ 
	= &\phi^{\mathbf{\tilde{u}}}(\barx{b}-1+p|t+1) +T_s v_\phi^{\mathbf{\tilde{u}}}(\barx{b}-1+p|t+1), \nonumber
	\intertext{and by recursively inserting the system dynamics, first \eqref{eq:SysDynBegin} and then \eqref{eq:SysDynvphi},}
	= &\phi^* (\barx{b}+1|t) + T_s \sum_{q = 1}^{p} v_\phi^{\mathbf{\tilde{u}}}(\barx{b}-1+q|t+1) \label{eq:Phiudidyn1} \\
	\begin{split} = &\phi^*(\barx{b}+1|t) + T_s \left( \sum_{q = 1}^{p} v_\phi^* (\barx{b}|t) \right. \\    &\qquad + T_s \left. \sum_{r=1}^{q} u^c(\barx{b}-2+r|t+1) \right).
	\end{split} \label{eq:Phiudidyn2}
	\end{align}

	Again, the first $N_p-b+1$ time steps are feasible because the candidate input sequence $\mathbf{\tilde{u}}(t+1)$ is the same as the previously feasible solution $\bs u^*(t)$. The main idea of proving feasibility of the remaining time steps is to bound the candidate joint distances by the previously feasible trajectory. 
	
	In the following, we need to consider three different cases, depending on the candidate input $u^c(k|t)$. In the first two cases, we compare the candidate trajectory directly to the previously feasible trajectory of the joint distances. In the third case, we show feasibility by induction.
	
	First, assume \mbox{$v_{\phi,i}^{\mathbf{\tilde{u}}}(\barx{b}-2+p|t+1) > T_s u_{\max}$} for some $p \in \mathbb{I}_{[1,b]}$ which implies 
	\begin{align*}
	v_{\phi,i}^{\mathbf{\tilde{u}}}(\barx{b}-2+r|t+1) &> T_s u_{\max} > 0
	\intertext{and}
	u^c_i(\barx{b}-2+r|t+1) &= -u_{\max} \\ 
	&\leq u^*_i(\barx{b}-1+r|t)
	\end{align*}
	for all $r \in \mathbb{I}_{[1,p]}$ by Lemma \ref{lemma:aux} and also 
	\begin{align*}
	v_{\phi,i}^{\mathbf{\tilde{u}}}(\barx{b}-1+p|t+1)>0
	\end{align*}
	due to Lemma \ref{lemma:aux}(ii). Then, inserting \eqref{eq:Phiudidyn2} yields
	\begin{align*}
	&\phi_i^{\mathbf{\tilde{u}}}(\barx{b}+p|t+1) \\
	= &\phi_i^*(\barx{b}+1|t) + T_s \sum_{q = 1}^{p} \Bigg( v_{\phi,i}^*(\barx{b}|t) \\ &\qquad + T_s \sum_{r=1}^{q} u_i^c(\barx{b}-2+r|t+1) \Bigg) \\
	\leq &\phi_i^*(\barx{b}+1|t)  + T_s \sum_{q = 1}^{p} \left( v_{\phi,i}^*(\barx{b}|t) + T_s \sum_{r=1}^{q} u^*_i(\barx{b}-1+r|t) \right) \\
	= &\phi_i^*(\barx{b}+1+p|t) \\
	\leq &\phi_{\max}
	\end{align*}
	and, by inserting \eqref{eq:Phiudidyn1},
	\begin{align*}
	&\phi_i^{\mathbf{\tilde{u}}}(\barx{b}+p|t+1) \\ = &\phi_i^*(\barx{b}+1|t) + T_s \sum_{q = 1}^{p} v_{\phi,i}^{\mathbf{\tilde{u}}}(\barx{b}-1+q|t+1) \\
	> &\phi_i^*(\barx{b}+1|t) \\
	\geq &-\phi_{\max}.
	\end{align*}
	
	Second, if $v_{\phi,i}^{\mathbf{\tilde{u}}}(\barx{b}-2+p|t+1) < -T_s u_{\max}$, the arguments are the same as above and are therefore omitted.
	
	Third, if $|v_{\phi,i}^{\mathbf{\tilde{u}}}(\barx{b}-2+p|t+1)| \leq T_s u_{\max}$ holds, which implies \begin{align*}
		u^c_i(\barx{b}-2+p|t+1) &= \linebreak -\frac{1}{T_s}v_{\phi,i}^{\mathbf{\tilde{u}}}(\barx{b}-2+p|t+1),
	\end{align*}
	then the joint distance $\phi_i^{\mathbf{\tilde{u}}}(\barx{b}+p|t+1)$ remains constant and, hence, feasible by induction:
	\begin{align*}
	&\phi_i^{\mathbf{\tilde{u}}}(\barx{b}+p|t+1) \\
	= &\phi_i^{\mathbf{\tilde{u}}}(\barx{b}-1+p|t+1) + T_s v_{\phi,i}^{\mathbf{\tilde{u}}}(\barx{b}-2+p|t+1) \\ &\qquad + T_s^2 u_i^c (\barx{b}-2+p|t+1) \\
	= &\phi_i^{\mathbf{\tilde{u}}}(\barx{b}-1+p|t+1).
	\end{align*}
	
	Combining these three results yields feasiblity for all elements of the sequence $\bs \phi^{\mathbf{\tilde{u}}}(t+1)$ except \linebreak for the last one. Hence, it only remains to show that the last element $\phi^{\mathbf{\tilde{u}}}(N_p|t+1)$ of the \mbox{sequence} is feasible. If $v_{\phi}^{\mathbf{\tilde{u}}} (N_p-1|t+1) = 0$ then \mbox{$\phi^{\mathbf{\tilde{u}}}(N_p|t+1)= \phi^{\mathbf{\tilde{u}}}(N_p-1|t+1)$} would be feasible as shown above. In the following, we show that indeed $v_{\phi}^{\mathbf{\tilde{u}}} (N_p-1|t+1) = 0$ holds, by a contradiction argument to feasibility of the previous solution $\bs u^*(t)$.
	
	First, assume otherwise $v_{\phi,i}^{\mathbf{\tilde{u}}}(N_p-1|t+1) \neq 0$ for some $i \in \mathbb{I}_{[1,N_l-1]}$. 

	If $v_{\phi,i}^{\mathbf{\tilde{u}}}(N_p-2|t+1) = 0$, then $u^c_i(N_p-2|t+1) = 0$ holds and thus $v_{\phi,i}^{\mathbf{\tilde{u}}}(N_p-1|t+1) = 0$, which is a contradiction. Hence, $v_{\phi,i}^{\mathbf{\tilde{u}}}(N_p-2|t+1) \neq 0$.
	
	If $|v_{\phi,i}^{\mathbf{\tilde{u}}}(N_p-2|t+1)| \leq T_s u_{\max}$, then, by \eqref{eq:Defuc}, \mbox{$u_i^c(N_p-2|t+1) = -\frac{1}{T_s} v_{\phi,i}^{\mathbf{\tilde{u}}}(N_p-2|t+1)$} which yields $v_{\phi,i}^{\mathbf{\tilde{u}}}(N_p-1|t+1) = 0$ due to \eqref{eq:SysDynvphi}.
	Again, this is a contradiction, so $|v_{\phi,i}^{\mathbf{\tilde{u}}}(N_p-2|t+1)| > T_s u_{\max}$.
	
	Moreover, if $v_{\phi,i}^{\mathbf{\tilde{u}}}(N_p-2|t+1) > T_s u_{\max}$ then this is also true for every time step before where $u_i^c(k|t+1)$ was applied, due to Lemma \ref{lemma:aux}. Furthermore, this implies \mbox{$v_{\phi,i}^{\mathbf{\tilde{u}}}(N_p-1|t+1)>0$} because of Lemma \ref{lemma:aux}(ii). Repeatedly inserting the system dynamics \eqref{eq:SysDynvphi} into this equation yields
	\begin{align*}
	&0 < v_{\phi,i}^{\mathbf{\tilde{u}}}(N_p-1|t+1) \\
	= &v_{\phi,i}^*(N_p-b-1|t) + T_s \sum_{q=1}^{b} u_i^c(N_p-b-2+q|t+1)  \\
	= &v_{\phi,i}^*(N_p-b-1|t) - T_s b u_{\max}  \\
	\leq &v_{\phi,i}^*(N_p-b-1|t) - v_{\phi,\max} \\
	\Rightarrow &v_{\phi,\max} < v_{\phi,i}^*(N_p-b-1|t),
	\end{align*}
	which is a contradiction to feasibility of the optimal solution at time $t$. Lastly, if $v_{\phi,i}^{\mathbf{\tilde{u}}}(N_p-2|t+1) < T_s u_{\max}$ the same contradiction to \mbox{feasibility} of the optimal solution is achieved by the same arguments. Hence, \mbox{$v_{\phi,i}^{\mathbf{\tilde{u}}}(N_p-1|t+1) = 0$} and $\phi^{\bs u_{di,i}}(N_p|k+1)= \phi^{\mathbf{\tilde{u}}}(N_p-1|t+1)$ which concludes the proof.
	\end{proof}
	
	\begin{remark} Note that in this proof we only investigated the constrained states $\phi(t)$, $v_\phi(t)$ and the input $u(t)$, which adhere to double integrator dynamics. Neither the additional states nor the cost function enter our analysis. This result can therefore be applied to any system with constraints on a \mbox{double} integrator subsystem, possibly additional unconstrained states and an arbitrary cost function, which is found in many different kinds of mobile robot applications. \end{remark}
	
	\subsection{Performance guarantees} \label{sec:performance}
	
	In this section, we briefly discuss performance guarantees for the proposed economic MPC scheme in terms of achieving a certain benchmark velocity. However, a thorough presentation of our results is out of the scope of this paper. The analysis in this section is tailored to an economic MPC scheme which aims to maximize the snake robot's forward velocity. In the following, the cost function in \mbox{Problem \ref{prob:EMPC}} is therefore chosen as \mbox{$J(x(t), \bs u(t)) = - \sum_{k=0}^{N_p} v_t(k|t)$}, as discussed in the previous sections. Hence, the results in this section would need to be adapted for a different cost function.
	
	In the literature, convergence to an optimal periodic orbit for economic MPC schemes is usually shown by exploiting certain dissipativity properties of the system; see for instance \cite{MUller2017_EconomicandDistributed,angeli2012,Muller2015}. Due to the complexity of the snake robot model, showing such a dissipativity seems not to be possible. 
	Instead, we assume existence of an auxiliary controller, which can sufficiently accelerate the snake robot, and yields a cost which is upper bounded.
	This is detailed below.
	Moreover, to properly state our result below, we additionally need to make the reasonable assumption that the economic MPC scheme will keep a certain velocity level once it has reached it. 

	More precisely, let $v_t(t)$ be the current velocity and let $\tilde{v}_t \in \mathbb{R}$ be a benchmark forward velocity used for performance characterization. 
	Then, we assume that the set $\{ x(t) \in \mathbb{X} | v_t(t) \geq \tilde{v}_t \}$ is forward invariant under the proposed economic MPC scheme.
	\begin{assumption} \label{assum:inv}
		Define $\mathcal{V} = \{ x(t) \in \mathbb{X} | v_t(t) \geq \tilde{v}_t \}$, where $\tilde{v}_t \in \mathbb{R}_{>0}$. Then, the set $\mathcal{V}$ is invariant in the economic MPC closed loop.
	\end{assumption}
	Moreover, for a given sampling time $T_s \in \mathbb{R}_{>0}$, denote by $\epsilon \in \mathbb{R}$ the solution of
	\begin{equation*}
	\begin{split}
	\epsilon := \max_{x(t) \in \mathbb{X}} \hspace{.5cm} & \frac{1}{T_s} (v_t(t) - v_t(t+1)).
	\end{split}
	\end{equation*}
	This constant $\epsilon$ is a measure for how much the snake robot can slow down from one time step to the next, i.e., \mbox{$v_t(t+1) \geq v_t(t)-T_s \epsilon$} holds for all $x(t) \in \mathbb{X}$. Note that it is independent of $u(t)$, since the input does not directly enter the dynamics of the forward velocity $v_t(t+1)$, if we investigate only one time step. 
	\begin{assumption} \label{assum:conv}
		Let $\tilde{v}_t \in \mathbb{R}_{>0}$. There exists some\footnote{A continuous function $\beta:\mathbb{R}_{\geq0}\times\mathbb{R}_{\geq 0} \rightarrow \mathbb{R}_{\geq0}$  is said to be of class $\mathcal{KL}$ if $\beta(\cdot,t)$ is strictly increasing, $\beta(0,t)=0$, and $\beta(r,\cdot)$ is decreasing with $\lim_{k\rightarrow\infty}\beta(r,k)=0$, cf.~\cite{Khalil2013}.} 
	$\beta\in\mathcal{KL}$,	such that for every $x(t) \in \mathbb{X}$ with $v_t(t) < \tilde{v}_t$ there exists an input sequence $\bs{\bar{u}}(t) \in \mathbb{U}^{N_p}$ which satisfies
	\begin{align*}
		\tilde{v}_t - v_t^{\bs{\bar{u}}}(t+k_1) \leq \beta(\tilde{v}_t - v_t(t), k_1), \\
		\sum_{i=1}^{k_2} v_t^{\bs{\bar{u}}}(t+i) \geq mv_t(t) - T_s \epsilon,
	\end{align*}
	for all $k_1 \in \mathbb{I}_{[0,N_p]}$ and for all $k_2 \in \mathbb{I}_{[1,N_p]}$.
	\end{assumption}
	Assumption \ref{assum:conv} can be justified by investigating the closed-loop trajectory of the forward velocity when applying the standard lateral undulation controller given by \eqref{eq:defcontLU}. \mbox{Combining} these two assumptions yields the desired performance result, i.e., convergence of the closed loop to the set $\mathcal{V}$. Hence, the application of the proposed economic MPC scheme leads to an acceleration sufficient to reach the velocity $\tilde{v}_t$ and, once it is achieved, remains greater than this benchmark velocity.
	\begin{lemma} \label{lemma:conv}
		Let Assumptions \ref{assum:inv} and \ref{assum:conv} hold and assume that Problem \ref{prob:EMPC} is initially feasible at time $t=t_0$ with $N_p > b$. Then, Problem \ref{prob:EMPC} is feasible for all $t \in \mathbb{I}_{\geq t_0}$ and the closed loop converges to the set $\mathcal{V}$.
	\end{lemma}
	The proof of Lemma \ref{lemma:conv} is omitted in this paper.
	All details can be found in \cite{Nonhoff2018}. Summarizing, we give performance guarantees for the economic MPC closed-loop snake robot system which are based on assumptions that can be justified, e.g., by simulations of existing standard controllers.
	
	\section{Numerical Simulations} \label{sec:sim}
	
	Having established recursive feasibility, we demonstrate the effectiveness of the proposed economic MPC scheme and compare it to a standard lateral undulation controller through numerical simulations. Our proposed ecnomic MPC algorithm was implemented in CasADi \cite{Andersson2018}. For the sake of comparability we choose to use a similar set of parameters as in \cite{Liljeback2012,Marafioti2014}, and \cite{Liljeback2013}. The snake robot's $N_l = 9$ links all have the same mass \mbox{$m=1$ kg}, the same length \mbox{$l=0.14$ m} and friction coefficients in normal and tangential direction $c_n=3$ and $c_t=1$, respectively. We set \mbox{$t_0=0$ s} and initialize the snake robot with \mbox{$\phi(0) = [0,0.01,-0.01,0.01,0,0,0.01,-0.01]^T$ m}, \mbox{$v_\phi(0) = 0 $ m/s}, \mbox{$v_t(0) = 0$ m/s} and \mbox{$v_n(0)=0$ m/s}. The sampling time is chosen as \mbox{$T_s = 0.05$ s} and the prediction horizon is set to $N_p=20$ which is equivalent to \mbox{$1$ s}. Finally, the constraint sets are given as
	\begin{align*}
	\phi_{\max} &= 0.052 \text{ m}, \\
	v_{\phi,\max} &= 0.109 \text{ m/s}, \\
	u_{\max} &= 0.2276 \text{ m/s\textsuperscript{2}}.
	\end{align*}
	
	In order to confirm recursive feasibility we verify the conditions of Lemma \ref{lemma:recfeas} and obtain $b = 10 < N_p = 20$ from \eqref{eq:defb}. Hence, the ecnomic MPC scheme is feasible for all $t \in \mathbb{I}_{>0}$, since it is easy to see that it is initially feasible at $t=0$ for $\bs u(0) = ( 0, \ldots, 0)$.
	
	For comparison, we consider the standard lateral \linebreak undulation controller (\ref{eq:defcontLU}) which is taken from \cite{Liljeback2012,Liljeback2013}. The lateral undulation gait pattern parameters and the tuning parameters of the controller are adjusted such that it fully utilizes the entire constraint sets once it periodically applies the lateral undulation gait pattern and reaches a high asymptotic \mbox{average} forward velocity. Therefore, the gait \linebreak pattern parameters are chosen as \mbox{$\alpha = 0.05$ m}, \mbox{$\omega = 120$ deg/s}, \mbox{$\delta = 40$ deg} and \mbox{$\phi_0 = 0$ m}. The tuning parameters of the controller are taken from \cite{Liljeback2012,Liljeback2013} and are set to \mbox{$k_d = 5$ 1/s} and \mbox{$k_p = 20$ 1/s\textsuperscript{2}}.
		
	\subsection{Economic MPC vs. Lateral Undulation}
	
	In the following, we first compare the velocity achieved by both controllers, i.e., the economic MPC scheme and the standard lateral undulation controller. Additionally, we highlight below the ability of our proposed economic MPC approach to incorporate arbitrary objectives in the optimization by accounting for energy consumption in the objective function, and again, compare it to the standard lateral undulation controller.
	
	\begin{figure}[t]	
		\setlength\fwidth{.9\columnwidth}
		\centering
		\begin{footnotesize}
		\input{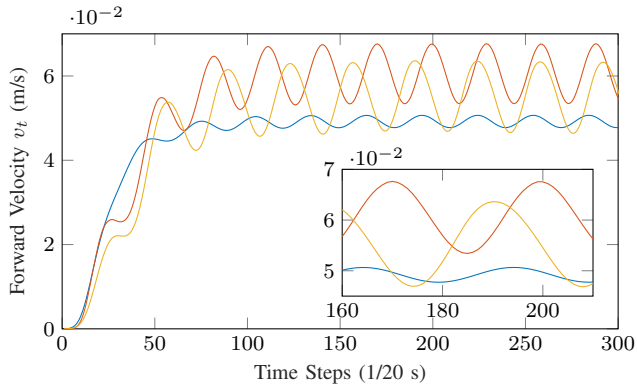}
		\end{footnotesize}
		\caption{Closed-loop trajectories of the forward velocity $v_t(t)$ for lateral undulation (blue) and the economic MPC scheme for values of $\gamma = 0$ (red) and $\gamma = 0.025$ (yellow).} \label{fig:velocititesEMPC}
	\end{figure}
	
	In the following, we set the cost function in the economic MPC optimization Problem~\ref{prob:EMPC} to $J(x(t), \bs u(t)) =\sum_{k=0}^{N_p} -v_t(k|t) + \gamma u^T(k|t) u(k|t)$, where $\gamma \in \mathbb{R}_{\geq 0}$. Here, the term $\gamma u^T(k|t) u(k|t)$ accounts for consumed energy, which is investigated below. We compare the proposed economic MPC algorithm for two different values of $\gamma$ with the standard lateral undulation controller. In the first experiment, we choose $\gamma =0$, which yields the cost function analyzed in Section \ref{sec:performance}, and in the second experiment, we choose $\gamma = 0.025$, which means that energy consumption is taken into account in the optimization.
	
	Figure \ref{fig:velocititesEMPC} shows the resulting closed-loop forward \mbox{velocity} for the lateral undulation controller given by (\ref{eq:defcontLU}) and the economic MPC scheme. One can identify a transient phase until approximately $100$ time steps and a periodic orbit afterwards. The asymptotic average velocity achieved by economic MPC for both values of $\gamma$ is clearly superior to lateral undulation. 
	
	This is achieved by a more efficient gait pattern: In \mbox{Figure \ref{fig:vphiEMPC}}, the trajectories of the joint velocities of joint $3$ are illustrated together with the corresponding constraint $\pm v_{\phi,\max}$ for both controllers. Surprisingly, a gait pattern different to lateral undulation emerges from economic MPC. Instead of a sinusoid as in lateral undulation, the proposed controller leads to an approximately piece-wise linear \mbox{trajectory} which results in the observed performance gain. Moreover, the economic MPC scheme achieves constraint satisfaction while the lateral undulation controller violates the constraints during the transient phase, which explains the qualitatively different closed-loop forward velocity during the first $50$ time steps.
	
	\begin{figure}[t]
		\setlength\fwidth{.9\columnwidth}
		\centering
		\vspace{9pt}
		\begin{footnotesize}
		  \hspace{-1.6em}
			\input{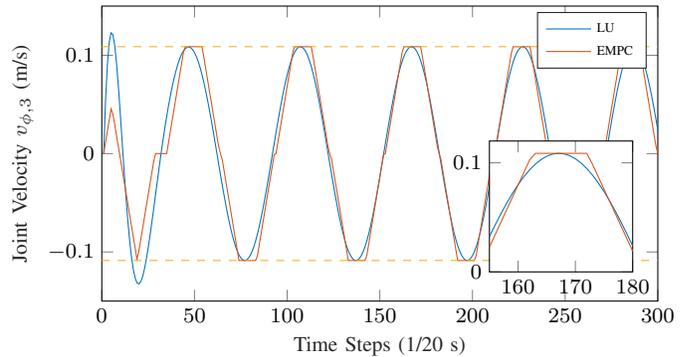}
		\end{footnotesize}
		\caption{Closed-loop trajectories $v_{\phi,i}(t)$ for $i = 3$ and constraints $\pm v_{\phi,\max}$ for $\gamma = 0$ and lateral undulation.} \label{fig:vphiEMPC}
	\end{figure}
	
	The superior velocity comes at the price of an increased energy consumption, which is investigated next. In order to compare the asymptotic average energy consumption and forward velocity, we computed the average energy consumption and velocity over the last $200$ time steps of our simulations. We consider $E=\sum_{t = 100}^{300} \frac{u^T(t)u(t)}{200}$ as a measure for the consumed energy by the controllers, while the asymptotic average velocity is computed by $v_{av} = \sum_{t=100}^{300} \frac{v_t(t)}{200}$. The results of this comparison are given in Table~\ref{tab}.
	
	\begin{table}[h]
		\caption{Comparison of the asymptotic average energy consumption and forward velocity}
		\label{tab}
		\begin{center}
			\begin{tabular}{|c|c|c|c|c|}
				\cline{2-5}
				\multicolumn{1}{c}{}& \multicolumn{2}{|c|}{Energy $E$} & \multicolumn{2}{c|}{Avg. velocity $v_{av}$} \\ 
				\hline
				Gait pattern & absolute & relative & absolute & relative\\
				\hline
				LU & $0.2072$ & $100.0 \%$ & $0.0494$ & $100 \%$ \\
				\hline
				EMPC, $\gamma=0$ & $0.2624$ & $126.6 \%$ & $0.0609$ & $123.3 \%$ \\
				\hline
				EMPC, $\gamma=0.025$ & $0.1979$ & $ 95.5 \%$ & $0.0554$ & $112.1 \%$ \\
				\hline
			\end{tabular}
		\end{center}
	\end{table}

	The data demonstrates that the economic MPC algorithm reaches a higher asymptotic average forward velocity in both cases, while it simultaneously leads to a lower energy consumption if the value of $\gamma$ is tuned to take energy consumption into account. Hence, the advantage of the proposed economic MPC scheme is twofold: It achieves a higher performance in terms of forward velocity as well as in terms of energy consumption by finding a superior gait pattern while ensuring constraint satisfaction.
	
	\subsection{Actuator fault}
	
	We next investigate the case of a joint failure occuring during operation to specifically illustrate the advantage of integrating the choice of the gait pattern into the closed loop. We simulate the scenario of joint $4$ blocking after \mbox{$10$ s}, and we therefore fix \mbox{$v_{\phi,4} \equiv 0$ m/s} afterwards. For this experiment, we set $\gamma = 0$ in the cost function and compare an economic MPC algorithm which has access to information about the failure, i.e., where we set \mbox{$v_{\phi,4} \equiv 0$ m/s} and \mbox{$u_4 \equiv 0$ m/s\textsuperscript{2}} in the dynamic constraints in Problem \ref{prob:EMPC}, with a fault-unaware economic MPC scheme. 
	
	Figure \ref{fig:FaultVelocities} shows the resulting closed-loop velocity for both controllers. It can be observed, that the fault-aware controller achieves a slightly higher asymptotic average velocity after the failure has occured. \mbox{Computing} the average velocity as in the previous subsection by \mbox{$v_{av} = \sum_{t = 300}^{500} \frac{v_t(t)}{200}$} yields \mbox{$0.0435$ m/s} for the algorithm with access to information about the failure and \mbox{$0.0418$ m/s} for the fault-unaware \linebreak algorithm.
	
	Moreover, in Figure \ref{fig:FaultPhi} the corresponding joint distances for joint $i = 3$ are illustrated. It becomes apparent that the economic MPC scheme adapts the gait pattern once the failure occurs in order to reach a higher forward velocity.
	
	\begin{figure}[t]
		\setlength\fwidth{.9\columnwidth}
		\centering
		\begin{footnotesize}
		\input{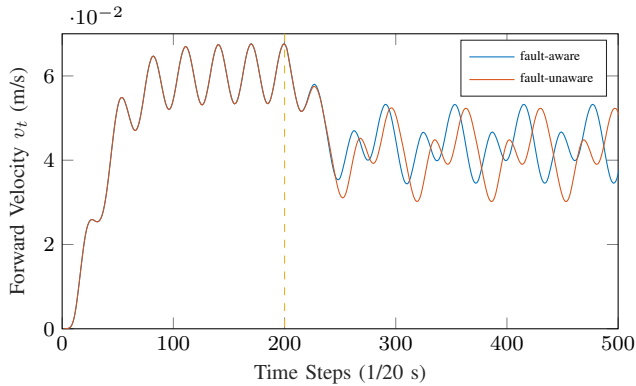}
		\end{footnotesize}
		\caption{Closed-loop trajectories $v_{t}(t)$ for the economic MPC scheme with and without information about the actuator fault.} \label{fig:FaultVelocities}
	\end{figure}

	\section{Conclusion}
	
	In this work, we proposed an economic MPC scheme for snake robot locomotion that integrates the choice of the gait pattern into the closed loop. We proved recursive feasibility of the economic MPC scheme in the closed loop and briefly sketched performance guarantees. We compared the proposed economic MPC algorithm to a standard lateral undulation controller from the literature and demonstrated that economic MPC is able to improve the snake robot's performance in terms of velocity and energy consumption at the same time, while ensuring constraint satisfaction.
	
	The theoretical foundations and insights gained by this work serve as a first step towards applying MPC to the control of snake robot locomotion with different \mbox{performance} criteria, subject to additional constraints and in the presence of \mbox{obstacles}. We believe that the ability to implicitly choose a gait pattern online through the economic MPC algorithm will be particularly beneficial for utilizing obstacles for locomotion, which is part of our ongoing work.
	
		\begin{figure}[t]
		\setlength\fwidth{.9\columnwidth}
		\centering
		\begin{footnotesize}
		\input{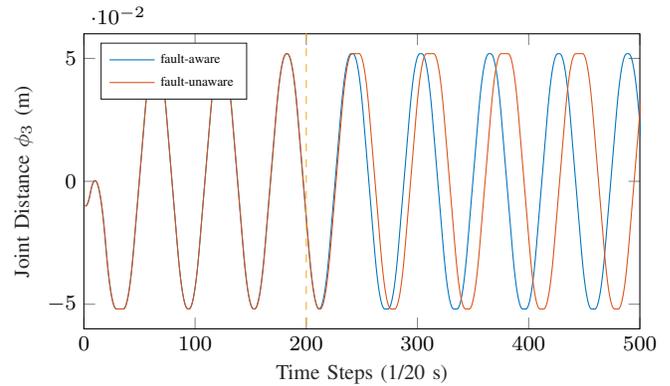}
		\end{footnotesize}
		\caption{Closed-loop trajectories $\phi_3(t)$ for the economic MPC scheme with and without information about the actuator fault.} \label{fig:FaultPhi}
	\end{figure}
	
	\bibliography{Bibliography}
	\bibliographystyle{ieeetr}
\end{document}